\documentclass[11pt,a4paper]{article}

\usepackage[truedimen,margin=34mm]{geometry} 

\usepackage{mathrsfs}
\usepackage{amssymb}
\usepackage{amsmath}
\usepackage{ascmac}
\usepackage{amsthm}
\usepackage[pdftex]{graphicx}
\usepackage[pdftex]{color}
\usepackage{booktabs}
\usepackage{setspace}
\usepackage{natbib}
\usepackage{color}
\usepackage{url}

\usepackage{titlesec}
\titleformat*{\section}{\large\bfseries}
\titleformat*{\subsection}{\it}

\newtheorem{thm}{Theorem}
\newtheorem{lem}{Lemma}

\def\be{{\beta}}

\def\ep{{\varepsilon}}

\def\Si{{\Sigma}}

\def\beh{{\widehat \beta}}

\def\Vh{\widehat{V}}

\def\bet{{\widetilde \be}}
\def\ept{{\widetilde \ep}}

\def\Si{{\Sigma}}
\def\Sih{{\widehat \Sigma}}

\def\psih{{\widehat \psi}}

\def\tr{{\rm tr}}

\def\E{{\rm E}}

\def\Re{{\mathbb R}}

\title{{\bf Improved Confidence Regions in Meta-analysis of Diagnostic Test Accuracy}}

\date{}
\author{}

\begin{document}

\maketitle
\doublespacing

\vspace{-1.5cm}
\begin{center}
{\large 
Tsubasa Ito$^1$ and Shonosuke Sugasawa$^2$
}
\end{center}

\noindent
$^1$M\&D Data Science Center, Tokyo Medical and Dental University\\
$^2$Center for Spatial Information Science, The University of Tokyo

\medskip
\begin{center}
{\bf \large Abstract} 
\end{center}

\vspace{-0cm}
Meta-analyses of diagnostic test accuracy (DTA) studies have been gathering attention in research in clinical epidemiology and health technology development, and bivariate random-effects model is becoming a standard tool. However, standard inference methods usually underestimate statistical errors and possibly provide highly overconfident results under realistic situations since they ignore the variability in the estimation of variance parameters. To overcome the difficulty, a new improved inference method, namely, an accurate confidence region for the meta-analysis of DTA, by asymptotically expanding the coverage probability of the standard confidence region. The advantage of the proposed confidence region is that it holds a relatively simple expression and does not require any repeated calculations such as Bootstrap or Monte Carlo methods to compute the region, thereby the proposed method can be easily carried out in practical applications. The effectiveness of the proposed method is demonstrated through simulation studies and an application to meta-analysis of screening test accuracy for alcohol problems.


\newpage
\section{Introduction}

Evidence synthesis methods have been gathering attention in diagnostic test accuracy (DTA) studies in clinical epidemiology and health technology development \citep{Lee2008}.
In this meta-analysis, the summary statistics in each study are two primary correlated outcomes of diagnostic, sensitivity and false positive rate ($1-{\rm specificity}$), and we are typically interested in summary receiver operating characteristic curve.
Moreover, DTA from different sources for studies are generally heterogeneous due to various factors, which should be adequately addressed to avoid underestimation of statistical errors and misleading conclusions \citep{Higgins2011}. 
Due to the potential correlations between two summary measures and potential heterogeneity, the bivariate random-effects models is adopted as the standard method for the meta-analysis \citep{Reitsma2005, Harbord2007}.

In the bivariate random-effects meta-analyses, standard inference methods depend on large sample approximations for the number of studies synthesized, for example the extended DerSimonian-Laird methods \citep{Chen2012, Jackson2010, Jackson2013} and restricted maximum likelihood (REML) estimation \citep{Reitsma2005, Jackson2011}, but the numbers of trials are often moderate or small in practice. 
In this situation, validity of the inference methods can be violated, which may lead over-confidence results, that is, coverage probabilities of the confidence regions or intervals cannot retain their nominal confidence levels and also the type-I error probabilities of the corresponding tests can be inflated. 
Such problem with random-effects models was well recognized in the context of both univariate and multivariate meta-analysis, even when the models are completely specified \citep{Veroniki2019}.
Recently, several refined methods have been proposed to improve confidence intervals in multivariate meta-analysis.
For example, \cite{Noma2018} developed improved confidence intervals in network meta-analysis using Bartlett-type corrections, and \cite{Noma2019} and \cite{Sugasawa2019} developed a unified method for computing accurate confidence intervals and regions in general random-effects meta-analysis.
However, these methods require computationally very intensive methods based on Monte Carlo or Bootstrap methods.
Also these methods considered confidence intervals or regions by inverting statistical hypothesis tests, thereby feasible ways to construct confidence regions are not necessarily obvious. 
On the other hand, there are a few analytical approaches to improve the standard approaches. 
\cite{Noma2011} and \cite{Guolo2012} considered higher order likelihood inference in the univariate meta-analysis, which cannot be directly applicable to more complicated multivariate meta-analysis. 
As more general approaches, \cite{Zucker2000} proposed an improved likelihood test in general linear mixed models through asymptotic expansions of the (restricted) maximum likelihood estimators, but the results include tedious algebraic expressions and are not useful in practice.

In this paper, we propose an improved confidence region for the bivariate random-effects meta-analysis for DTA, which does not require any repeated calculation methods and has relatively simple analytical expressions, thereby the proposed method could be easily employed in practical applications. 
The key mathematical tool is the distributional properties between the ordinary least squares estimator and residuals, and define a class of estimators of variance parameters in random-effects models.
Then, we find a relatively simple formula for asymptotic approximation of the coverage probability of the crude Wald-type confidence intervals and regions, and construct a second order accurate confidence region.
We carry out extensive simulation studies to compare the performance of the proposed confidence region with that of the standard REML method, and demonstrate that the proposed method shows quite reasonable empirical coverage than REML while the computational cost in both methods are almost identical. 
We also demonstrate the proposed method through an application to meta-analysis of screening test accuracy for alcohol problems.

This paper is set out as follows. 
In Section \ref{sec:method}, we describe the proposed confidence region under bivariate random-effects models. 
In Section \ref{sec:num}, we numerically demonstrate the proposed confidence region together with existing methods through extensive simulation studies and an application with real dataset.
We conclude with a short discussion in Section \ref{sec:Disc}.
R code implementing the proposed method is available at GitHub repository (\url{https://github.com/sshonosuke/CCR-BMA}).

\section{Improved Confidence Regions in Meta-analysis for Diagnostic Test Accuracy}\label{sec:method}

\subsection{Bivariate random-effects models and confidence region}
There has been increasing interest in systematic reviews and meta-analyses of data from diagnostic accuracy studies.
For this purpose, a bivariate random-effect model \citep{Reitsma2005, Harbord2007} is widely used.
Following \cite{Reitsma2005}, we define $\mu_{Ai}$ and $\mu_{Bi}$ as the logit-transformed true sensitivity and specificity, respectively, in the $i$th study. 
Let $y_{Ai}$ and $y_{Bi}$ be the observed logit-transformed sensitivity and specificity, and $s_{Ai}$ and $s_{Bi}$ are associated standard errors.
The bivariate model assumes that $\mu_i=(\mu_{Ai},\mu_{Bi})^t$ and $y_i=(y_{Ai},y_{Bi})^t$ follow bivariate normal distributions:
\begin{equation}\label{BMA}
y_i|\mu_i\sim N_2(\mu_i,S_i), \ \ \ \ \mu_i\sim N_2(\beta,\Si), \ \ \ \ i=1,\ldots,n,
\end{equation}
where $\beta=(\beta_A,\beta_B)^t$ is a vector of the average logit-transformed sensitivity and specificity, and $S_i={\rm diag}(s_{Ai},s_{Bi})$.
Note that there is no correlation between $y_{Ai}$ and $y_{Bi}$ given $\mu_i$ since sensitivity and specificity are calculated based on individuals identified as positive and negative, respectively.   
Here $\Si$ is unstructured, so that it allows correlation between $\mu_{Ai}$ and $\mu_{Bi}$.
Let $y=(y_1^t,\ldots,y_n^t)^t\in \Re^{2n}$, $X=(X_1,\ldots,X_n)^t$ with $X_i=I_2$ and $S={\rm diag}(S_1,\ldots,S_n)$.
Then, the model (\ref{BMA}) is equivalent to $y\sim N(X\beta,I_n\otimes \Sigma+S)$.

Our primary interest is a confidence region of $\beta$.
Hence, the variance-covariance matrix $\Sigma$ is a nuisance parameter. 
These parameters are typically estimated via (restricted) maximum likelihood methods based on the model assumption (\ref{BMA}). 
For summarizing the results of the meta-analysis, we typically employ confidence region of $\beta$ rather than separate confidence intervals since sensitivity and specificity could be highly correlated.
\cite{Reitsma2005} suggested the $100(1-\alpha)\%$ confidence region for $\beta$ as the interior points of the ellipse defined as
\begin{equation}\label{NCR}
\left\{\beta : (\beh(\Sih)-\beta)^tV(\Sih)^{-1}(\beh(\Sih)-\beta)\leq \chi_2^2(\alpha)\right\},
\end{equation}
where $\beh(\Sigma)$ is the generalized least squares estimator of $\beta$ and $\Sih$ is the restricted maximum likelihood estimator of $\Sigma$, $V(\Si)=\{X^t(I_n\otimes \Sigma+S)X\}^{-1}$ is the variance-covariance matrix of $\beh$, $\Sih$ is the restricted maximum likelihood estimator and $\chi_2^2(\alpha)$ is the upper $100\alpha\%$ point of the $\chi^2$ distribution with $2$ degrees of freedom.
The joint confidence region  (\ref{NCR}) is approximately valid, that is, the coverage error converges to $1-\alpha$ as the number of studies $n$ goes to infinity.
However, when $n$ is not sufficiently large, the coverage error is not negligible, and the region (\ref{NCR}) would under-cover the true $\beta$.

\subsection{Improved confidence region}
In this work, we derive an improved confidence region whose coverage error is $o(n^{-1})$, which has higher order accuracy than the standard confidence region (\ref{NCR}).
The main idea is to derive an approximation formula of the coverage probabilities of the confidence region of the form (\ref{NCR}) with a certain class of estimators for $\Sigma$, and derive an improved confidence region in an analytical form.

We consider a class of estimators $\Sih(y)$ satisfying the following conditions:
\begin{itemize}
\item[(C1)]
$\Sih$ is an even function of $y$ and translation invariant, that is, $\Sih(y)=\Sih(-y)$, and $\Sih(y+c)=\Sih(y)$ for any $c\in\Re^{2n}$.

\item[(C2)]
$\Sih$ is $\sqrt n$-consistent and $\Sih$ is second-order unbiased, namely $\Sih-\Sigma=O(n^{-1/2})$ and $\E[\Sih]=\Sigma+o(n^{-1})$.

\item[(C3)]
$\Sih$ is a function of $Py$ with $P=I_{2n}-X(X^tX)^{-1}X^t$. 
\end{itemize}

The first condition (C1) is typically satisfied by typical estimators including (restricted) maximum likelihood estimator and moment-based estimators.
The $\sqrt{n}$-consistency in (C2) is also a standard condition, but second order unbiasedness of $\psih$ is not always satisfied.
For example, the maximum likelihood (ML) estimator does not necessarily hold the property.
The condition (C3) requires that the estimator should be function of residuals based on ordinary least squares estimator of $\beta$, which is a key assumption in constructing the proposed confidence region.
The condition (C3) enables us to get a relatively simple form of the corrected confidence region.
Note that the typical estimators (e.g. REML) does not satisfy the condition (C3).
As a specific estimator satisfying all the above conditions, we employ the following moment-based estimator:
$$
\Sih_0=
\frac1n \sum_{i=1}^n \left\{(y_i-X_i\beh^{{\rm OLS}})(y_i-X_i\beh^{{\rm OLS}})^t-S_i\right\},
$$
where $\beh^{{\rm OLS}}=(X^tX)^{-1}X^ty$ is the ordinary least squares estimator.
Since this estimator is not second-order unbiased, let $\Sih$ be a bias corrected version, that is, $\Sih=\Sih_0-{\rm Bias}_{\Sih_0}(\Sih)$ with ${\rm Bias}_{\Sih_0}(\Si)=-n^{-2} \sum_{i=1}^n(\Sigma+S_i)$, which satisfies all the conditions (C1)$\sim$(C3).  
We also note that given the estimator of $\Sigma$, the parameter $\beta$ can be estimated via the generalized least squares estimator given by 
$$
\beh(\Sigma)=\{X^t(I_n\otimes \Sigma+S)^{-1}X\}X^t(I_n\otimes \Sigma+S)^{-1}y.
$$

In order to improve the coverage accuracy of the confidence region (\ref{NCR}), we consider a class of confidence regions of the form
\begin{equation}
\label{CCR}
\bigg\{\beta: (\beh(\Sih)-\be)^tV(\Sih)^{-1}(\beh(\Sih)-\be)\leq x(1+h(\Sih))\bigg\},
\end{equation}
where $h(\cdot)$ is a function with order $O(n^{-1})$.
When $h(\Sigma)=0$ and $x=\chi_2^2(\alpha)$, the confidence region (\ref{CCR}) reduces to (\ref{NCR}), thereby the function $h$ can be regarded as an adjustment function to achieve reasonable coverage properties.
If $\Sih$ satisfies the conditions (C1)$\sim$(C3), the approximation formula of coverage probability of the confidence region (\ref{CCR}) can be obtained in a relatively simple form, as summarized in the following theorem.

\begin{thm}\label{thm:cp}
Suppose that $\Sih$ satisfies the conditions (C1)$\sim$(C3), and $h\equiv h(\Sigma)$ is a function with order $O(n^{-1})$.
Then, it follows that 
\begin{align*}
&P\big\{(\beh(\Sih)-\be)^tV(\Sih)^{-1}(\beh(\Sih)-\be)\leq x(1+h)\big\}\\
& \ \ \ \ =F_{k}(x)+hxf_k(x)+\left(\frac{B_1}{4}-\frac{B_2}{2}+2B_3\right)f_{k+2}(x)\\
& \ \ \ \ \ \ \  -\left(\frac{B_1}{4}+\frac{B_2}{2}\right)f_{k+4}(x)+O(n^{-3/2}),
\end{align*}
where $F_{k}(\cdot)$ and $f_k(\cdot)$ are the cumulative distribution and density function of the chi-squared distribution with degrees of freedom $k$, respectively, and $B_1, B_2$ and $B_3$ are $O(n^{-1})$ quantities given by  
\begin{equation}
\label{B123}
B_1=\E\Big[\tr (K(\Sih,\Sigma))^2\Big], \ \ \ \ 
B_2=\tr\Big(\E[K(\Sih,\Sigma)^2]\Big),  \ \ \ \ 
B_3=\tr\Big(\E[K(\Sih,\Sigma)]\Big),
\end{equation}
with $K(\Sih,\Sigma)=\big\{V(\Sih)-V(\Sigma)\big\}V(\Sigma)^{-1}$.
\end{thm}

\bigskip
From Theorem \ref{thm:cp}, It turned out that the coverage probability of the confidence region (\ref{CCR}) is a simple functional of the djustment function $h(\cdot)$.
Hence, to achieve higher accuracy of the confidence region, it suffices to choose $h(\cdot)$ such that 
$$
hxf_k(x)+\left(\frac{B_1}{4}-\frac{B_2}{2}+2B_3\right)f_{k+2}(x)-\left(\frac{B_1}{4}+\frac{B_2}{2}\right)f_{k+4}(x)=0.
$$ 
Since $f_{k+2}(x)/f_k(x)=x/k$ and $f_{k+4}(x)/f_k(x)=x^2/k(k+2)$, the solution with respect to $h$ is given by 
\begin{equation}
\label{h-func}
h(\Sigma)=\frac1k\left(\frac{B_1}{4}-\frac{B_2}{2}+2B_3\right)-\frac{x}{k(k+2)}\left(\frac{B_1}{4}+\frac{B_2}{2}\right).
\end{equation}
We also note that $h(\Sih)=h(\Sigma)+o_p(n^{-1})$ since $h(\cdot)=O(n^{-1})$.
Then, the confidence region given in (\ref{CCR}) with $h(\cdot)$ given in (\ref{h-func}) holds the second-order accuracy as shown in the following theorem.

\begin{thm}\label{thm:ccr}
Let ${\rm CCI}_{\alpha}$ be the confided region of the form (\ref{CCR}) with $h(\cdot)$ given in (\ref{h-func}).
Then, it follows that $P(\beta\in {\rm CCI}_{\alpha})=1-\alpha+o(n^{-1})$.
\end{thm}

It is notable that the derived confidence region has analytical expressions, so that it does not require any computationally intensive methods such as bootstrap and Monte Carlo integration as used in \cite{Sugasawa2019} and \cite{Noma2019}.
For practical implementation, we need to obtain the expressions of $B_1$, $B_2$ and $B_3$ given in (\ref{B123}).
We here provide approximation formulas.
we can obtain $B_\ell^{\ast}, \ \ell=1,2,3$ which satisfy $B_\ell^{\ast}=B_\ell+o(n^{-1})$, where 
\begin{equation}\label{B-ast}
\begin{split}
B_1^{\ast}
=&
\frac2{n^2}\sum_{i=1}^n\sum_{j=1}^n\sum_{k=1}^n \tr\Big(VU_{jik}VU_{kij}\Big),\\
B_2^{\ast}
=&
\frac1{n^2} \sum_{i=1}^n\tr\Big(V\sum_{j=1}^nU_{jij}^2\Big)
+\frac1{n^2} \sum_{i=1}^n\sum_{j=1}^n\sum_{k=1}^n\tr^2\Big(VU_{ijk}\Big),\\
B_3^{\ast}
=&
B_2-\frac1{n^2}\sum_{i=1}^n\sum_{j=1}^n \tr\Big(VU_{iji}D_jD_i^{-1}\Big)\\
&-\frac1{n^2}\sum_{i=1}^n\sum_{j=1}^n \tr\Big(D_i^{-1}D_j\Big)\tr\Big(VU_{iji}\Big),
\end{split}
\end{equation}
for $U_{ijk}=D_i^{-1}D_jD_k^{-1}$ and $D_i=\Sigma+S_i$. 
The detailed derivation is given in the Supplementary Material.
Note that using $B_\ell^{\ast}$ instead of $B_\ell$ in the derived confidence region does not change the coverage accuracy shown in Theorem \ref{thm:ccr} since the difference between $B_\ell^{\ast}$ and $B_{\ell}$ is only $o(n^{-1})$.

\section{Numerical Studies}\label{sec:num}

\subsection{Simulation study}\label{sec:sim}
We carried out extensive simulation studies to assess the finite sample performance of the proposed confidence region (\ref{CCR}) together with the approximate confidence region (\ref{NCR}) by \cite{Reitsma2005}.
In this study, we do not consider possible competitors by \cite{Noma2019, Sugasawa2019} due to two reasons; the coverage performance has been already confirmed in their papers, and calculation of sizes of their confidence regions are so intensive that it is not feasible to repeatedly calculate them in our simulation study. 
Hence, the following simulation study is supposed to compare the performance of the proposed and standard methods, both of which have almost the same computational time.

In the model (\ref{BMA}), we set $\beta=(0,0)$ and $\Sigma_{11}=\Sigma_{22}=\tau^2$ and $\Sigma_{12}=\Sigma_{21}=\tau^2\rho$.
We considered 8 scenarios of the between study variances $\tau^2\in \{0.1, 0.2, \ldots, 0.8\}$ and 5 scenarios of the between study correlations $\rho\in \{0, 0.2, \ldots, 0.8\}$.
Following, \cite{Jackson2014}, for each simulation, two within-study variances $s_{Ai}$ and $s_{Bi}$ were simulated from a scaled chi-squared distribution with 1 degree of freedom, multiplied by 0.25, and truncated to lie within the interval $[0.009, 0.6]$, so the expected values of the variance is $0.20$. 
We changed the number of studies $n$ over 8,16 and 24, and set the nominal level $\alpha$ to $0.05$.
Based on 1000 replications, we evaluated empirical coverage probabilities of 95\% confidence regions of the true parameters vector $\beta$ obtained from the proposed corrected (CCR) method as well as the standard naive (NCR) method. 
For simplicity, we evaluated coverage rates assessing rejection rates of the test of null hypothesis for the true parameters.
Since areas of the corrected confidence region is approximately $1+h(\Sih)$ times larger than those of naive ones, we also computed median values of $h(\Sih)$ among 1000 replications.  
To see the degree of heterogeneity depending on $n$ and $\tau^2$, we computed heterogeneity measure given by $I^2=\tau^2/(Q+\tau^2)$ with $Q=(n-1)\sum_{i=1}^nw_i/\{(\sum_{i=1}^nw_i)^2-\sum_{i=1}^nw_i^2\}$ and $w_i=s_{Ai}^{-1}$ in each iteration, which were averaged over 1000 replications. 
Note that $I^2\in (0,1)$ and lager value of $I^2$ indicates more significant heterogeneity in the data.

The averaged values of $I^2$ are reported in Table \ref{tab:I2}, which indicates that our simulation scenarios contain a wide range of heterogeneity. 
The obtained coverage probabilities and the median values of $h(\Sih)$ are shown in Figures \ref{fig:cp} and \ref{fig:H}, respectively. 
From Figure \ref{fig:cp}, it is observed that the simulated coverage probabilities of the standard NCR seriously smaller than the nominal level ($95\%$), especially in the case with the small number of studies ($n=8$), possibly because of the naive approximation in (\ref{NCR}).
On the other hand, the proposed CCR provides considerably better performance than NCR as the coverage probabilities are relatively close to the nominal level. 
Although the coverage probability of CCR tend to be larger than the nominal level when $\tau$ is small and/or $\rho$ is large, such a conservative property would be much more desirable than the over-confident property that NCR shows.  
From Figure \ref{fig:H}, we can see that the area of CCR is much larger than that of NCR since CCR takes account of additional variability due to the estimation of the variance-covariance matrix, so it is quite reasonable that $h(\Sih)$ decreases as $n$ increases.
Moreover, we can also observe that the areas of CCR decreases as $\tau^2$ increases and increases as $\rho$ increases, which are consistent to the results of the overage probabilities shown in Figure \ref{fig:cp}.

\begin{table}[htp!]
\caption{Averaged values of the heterogeneity measure $I^2 (\%)$ based on 1000 replications.
\label{tab:I2}
}
\centering
\vspace{0.5cm}
\begin{tabular}{ccccccccccccccccccc}
\hline
 & & \multicolumn{8}{c}{$\tau^2$}\\
$n$ &  & 0.1 & 0.2 & 0.3 & 0.4 & 0.5 & 0.6 & 0.7 & 0.8\\
\hline
8 &  & 27.2 & 63.4 & 75.1 & 81.4 & 85.0 & 87.2 & 89.1 & 90.8 \\
16 &  & 49.0 & 74.9 & 83.1 & 87.4 & 89.7 & 91.6 & 92.7 & 93.7 \\
24 &  & 56.2 & 78.0 & 85.4 & 88.9 & 91.2 & 92.7 & 93.7 & 94.6 \\
\hline
\end{tabular}
\end{table}

\begin{figure}[htp!]
\centering
\includegraphics[width=12cm,clip]{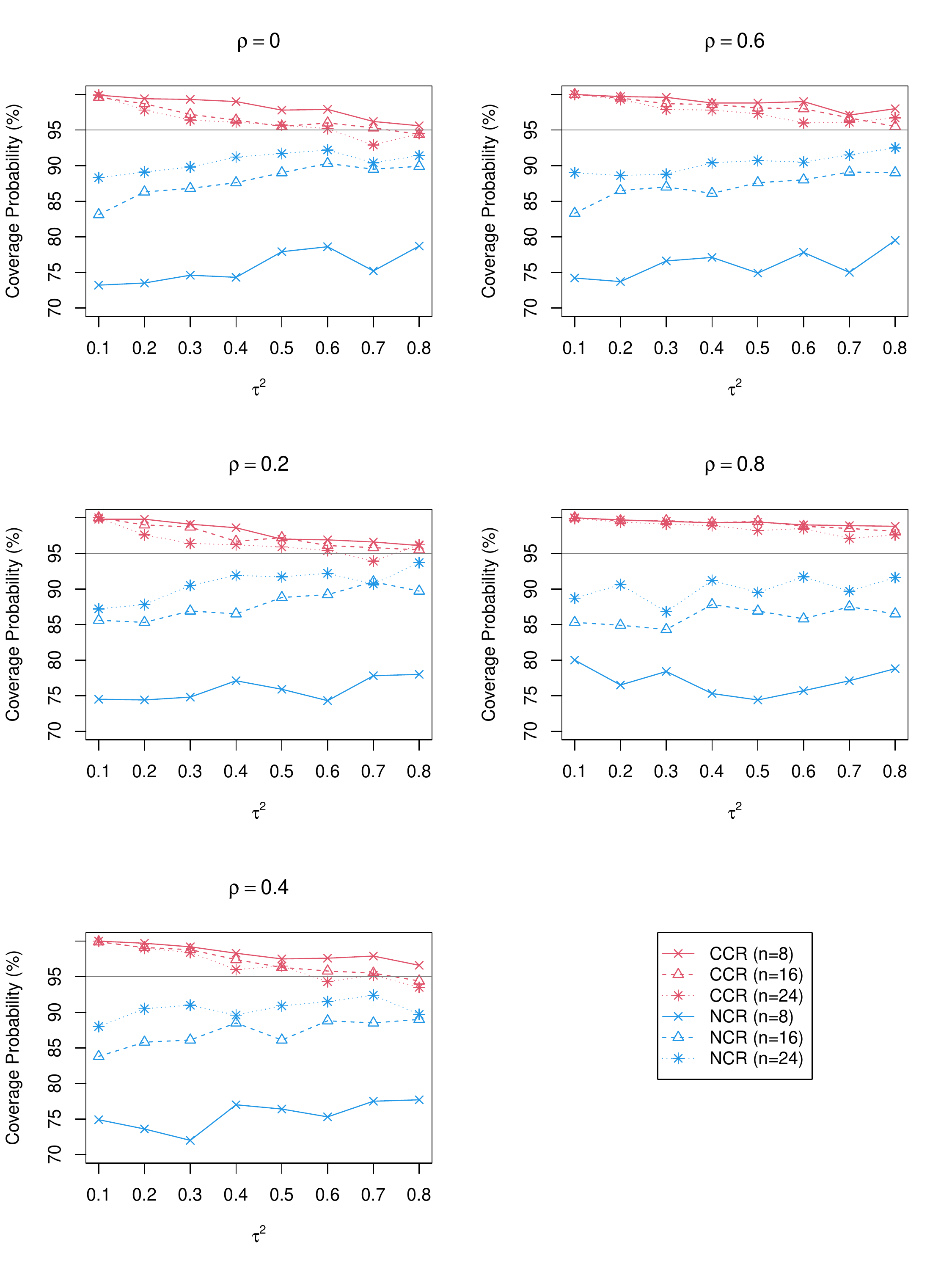}
\caption{The coverage probabilities of the proposed CCR and NCR based on 1000 replications under various combinations of $\tau^2, \rho$ and $n$.}
\label{fig:cp}
\end{figure}

\begin{figure}[htp!]
\centering
\includegraphics[width=12cm,clip]{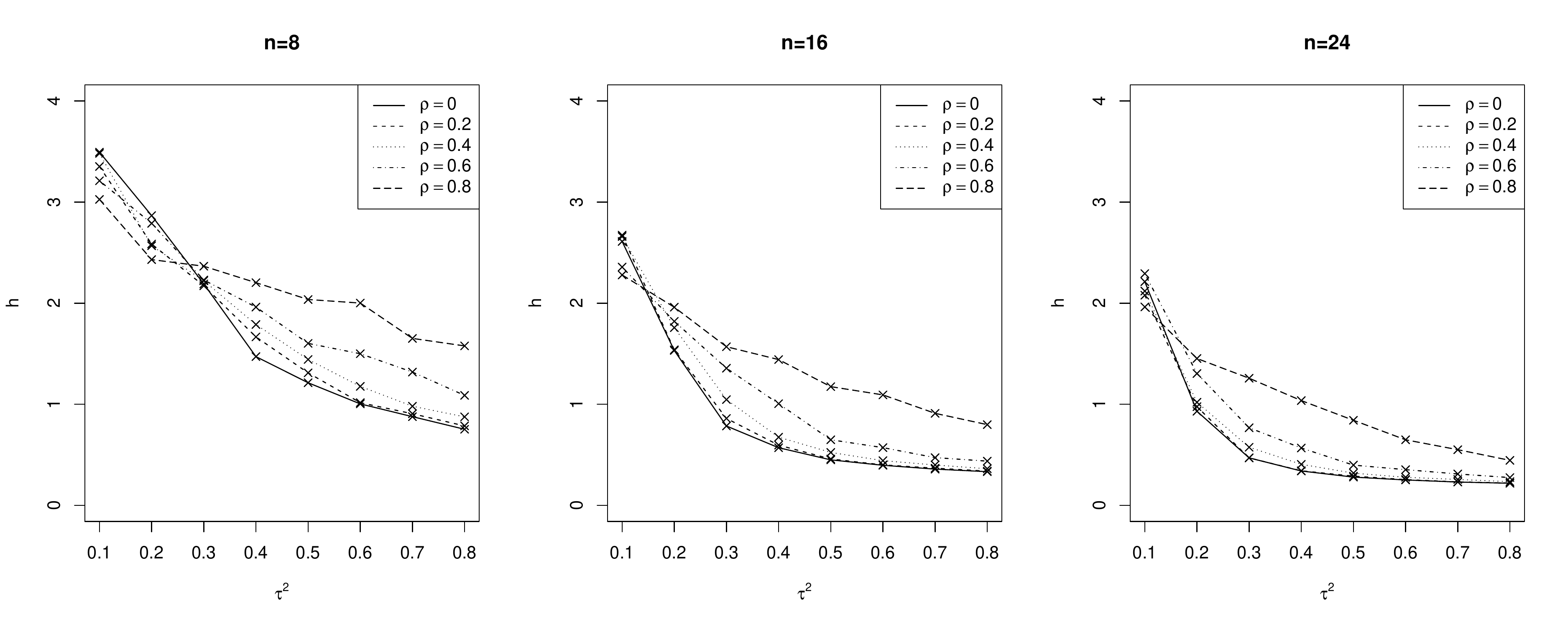}
\caption{The median values of $h(\Sih)$ in the proposed CCR based on 1000 replications under various combinations of $\tau^2, \rho$ and $n$. Note that $h(\Sih)=0$ means that the area of CCR is the same as that of NCR. 
}
\label{fig:H}
\end{figure}

\subsection{Example: screening test accuracy for alcohol problems}\label{sec:BMA-app}

Here we provide a re-analysis of the dataset given in \cite{Kriston2008}, including $n=14$ studies regarding a short screening test for alcohol problems.
Following \cite{Reitsma2005}, we used logit-transformed values of sensitivity and specificity, denoted by $y_{Ai}$ and $y_{Bi}$, respectively, and associated standard errors $s_{Ai}$ and $s_{Bi}$.
For the bivariate summary data, we first fitted the bivariate models (\ref{BMA}) using the restricted maximum likelihood method and found that $\widehat{\rho}=0.854$, and the heterogeneity measure $I^2$ for sensitivity and specificity are respectively given by $94.8\%$ and $98.9\%$, so there seems notable heterogeneity in the data. 
We then computed $95\%$ CRs of $\beta$ based on NCR (\ref{NCR}) given in \cite{Reitsma2005} and the proposed CCR.
Following \cite{Reitsma2005}, the obtained two CRs of $\beta$ were transformed to the scale $({\rm logit}(\beta_A),1-{\rm logit}(\beta_B))$, where ${\rm logit}(\beta_A)$ and $1-{\rm logit}(\beta_B)$ are the sensitivity and false positive rate, respectively. 
The obtained two CRs are presented in Figure \ref{fig:Audit} with a plot of the observed data, summary points $\beh$, and the summary receiver operating curve. 
The approximate CR is smaller than the proposed CR, which may indicate that the approximation method underestimates the variability of estimating nuisance variance parameters. 
In Figure \ref{fig:Audit}, we also reported the confidence region based on \cite{Sugasawa2019} using Monte Carlo simulation to compute accurate $p$-values of likelihood ratio statistics. 
The two regions based on the proposed method and \cite{Sugasawa2019} are slightly different but both are clearly wider than the naive confidence region. 
On the other hand, the computation time of the proposed method was less than 1 second while the inference method by \cite{Sugasawa2019} took more than 12 hours, where the program was run on a PC with a 3 GHz 8-Core Intel Xeon E5 8 Core Processor with approximately 16GB RAM.

\begin{figure}[h]
\centering
\includegraphics[width=11cm,clip]{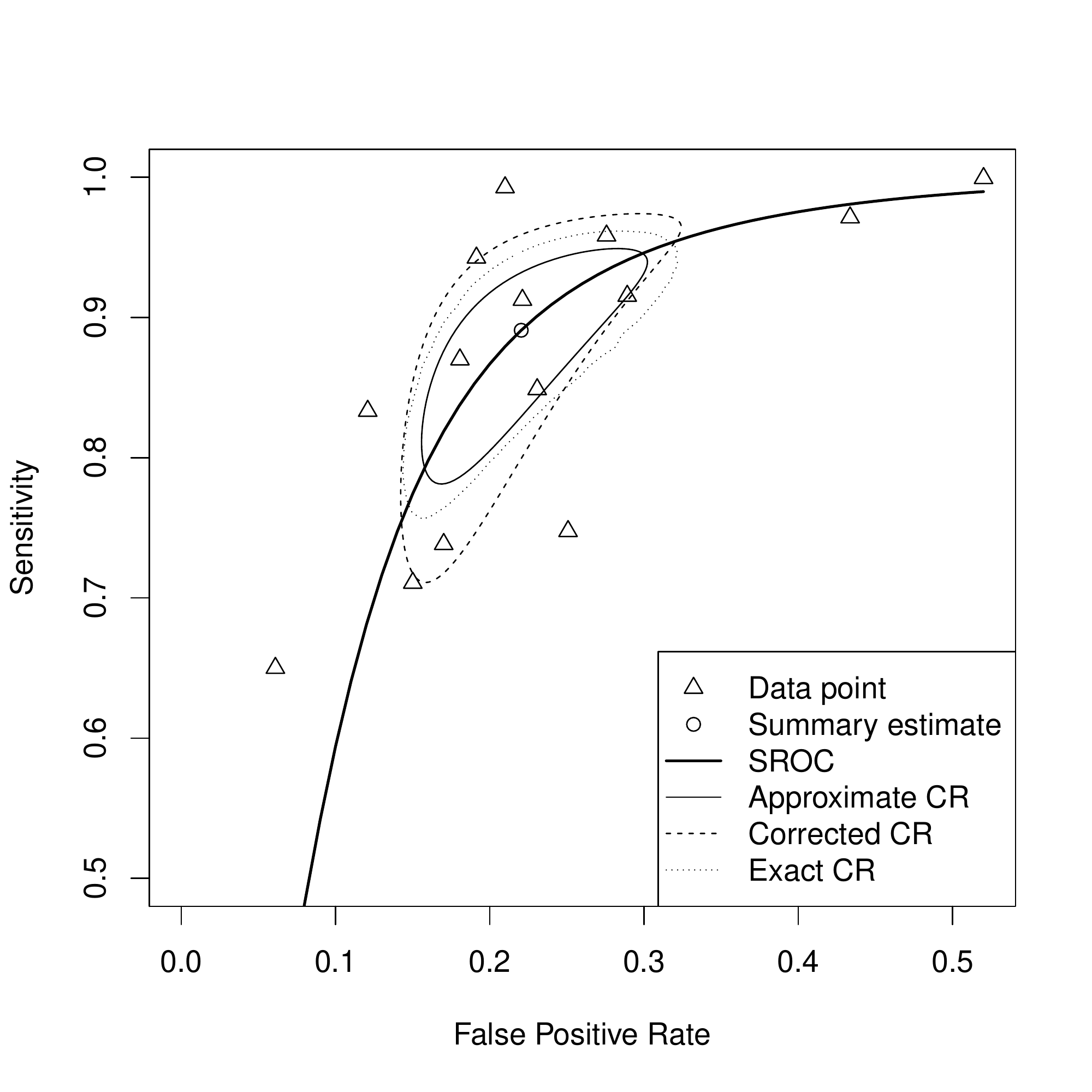}
\caption{The approximate and corrected CRs and summary receiver operating characteristics (SROC) curve.}
\label{fig:Audit}
\end{figure}

\section{Discussion}\label{sec:Disc}
In this paper, we presented an improved confidence region for random effects meta-analysis for diagnostic test accuracy without using repeated calculations such as Monte Carlo or Bootstrap methods. 
The proposed confidence region has relatively simple form and they are shown to have second order accurate coverage probability while the standard inference methods (e.g. REML) have significant coverage errors.
In simulation studies, we demonstrated that possible under-coverage properties of the standard methods under the small number of studies to be synthesized while the proposed method provides reasonable coverage properties.

A possible limitation of the proposed method might be that the coverage accuracy still depends on the number of studies. 
On the other hand, inference methods that does not rely on large sample approximation have been recently proposed \citep[e.g.][]{Noma2019, Sugasawa2019}, which are computationally intensive, so they would not be necessarily practical.
Then, the proposed method would be regarded as a reasonable compromise between methods with exact empirical coverage and computational efficiency.

\section*{Acknowledgments}
This research was supported by Japan Society of Promotion of Science KAKENHI (grant number: 18K12757).

\bibliographystyle{chicago}
\bibliography{refs}

\newpage

\begin{center}
{\LARGE \bf Supplementary Material for ``Improved Confidence Regions in Meta-analysis of Diagnostic Test Accuracy"}

\medskip
\begin{center}
Tsubasa Ito$^1$ and Shonosuke Sugasawa$^2$
\end{center}

\noindent
$^1$Department of Statistical Thinking, The Institute of Statistical Mathematics\\
$^2$Center for Spatial Information Science, The University of Tokyo
\end{center}

\setcounter{equation}{0}
\renewcommand{\theequation}{S\arabic{equation}}
\setcounter{section}{0}
\renewcommand{\thesection}{S\arabic{section}}
\setcounter{lem}{0}
\renewcommand{\thelem}{S\arabic{lem}}
\setcounter{page}{1}

\vspace{5mm}
This supplementary material provides the proofs and the detailed derivations of Theorem 1, Theorem 2, equation (6).
In what follows, we denote $\Si_n=I_n\otimes \Si$ and write $\beh(\Sigma)$ as $\bet$, $\beh(\Sih)$ as $\beh$, $\Si_n(\psi)$ as $\Si_n$, $V(\Sigma)$ as $V$ and $V(\Sih)$ as $\Vh$, for notational simplicity.

\section{Key lemmas}

We first introduce lemmas which play important roles in the proof of Theorems \ref{thm:cp}.
The first lemma is used for deriving the conditional distribution of $\beh$.
\begin{lem}\label{lem:1}
Under the conditions {\rm (C1)}-{\rm (C3)} given in the main document, $\bet$ is independent of $Py$ for $P=I_n-X(X^tX)^{-1}X^t$.
Also, $\beh-\bet$ is a function of $Py$, and independent of $\bet$.
\end{lem}

\begin{proof}
Let $\ept=y-X\be$, which is distributed as $N(0, \Sigma_n+S)$. 
Since $\bet-\be=VX^t(\Sigma_n+S)^{-1}\ept$, it holds that 
\begin{align*}
\E[Py(\bet-\be)^t] V^{-1} 
&=P\E[\ept\ept^t](\Si_n+S)^{-1}X=PX=0.
\end{align*}
Since $V$ is a full-rank matrix, we have $\E[Py(\bet-\be)^t]=0$, that is the covariance of $Py$ and $\bet$ is $0$, which implies that $\bet$ is independent of $Py$ from the normality assumption.
Now, we write $\bet$ as $\bet(\Si,y)$ and $\beh$ as $\beh(\Sih(y),y)$.
Since $\bet(\Si,y+XT)=\bet(\Si,y)+T$ and $\beh(\Sih(y+XT),y+XT)=\beh(\Sih(y),y)+T$ from (C3), we have
\begin{align*}
\beh(\Sih(y+XT),y+XT)-\bet(\Si,y+XT)
=\beh(\Sih(y),y)-\bet(\Si,y),
\end{align*}
which implies that $\beh-\bet$ is invariance with respect to the translation $y\to y+XT$.
Moreover, $Py$ is maximal invariant with respect to the translation $y\to y+XT$ since $P(y+XT)=Py$ and $Py_1=Py_2$ implies that $y_1=y_2+XT'$ for $T'=(X^tX)^{-1}X^t(y_1-y_2)$.
Then, $\beh-\bet$ is a function of $Py$ from Theorem 2 in \cite{Berger1985}, p.403.
\end{proof}

In the next lemma, we show the first order bias of the plug-in estimator $\Vh$ is approximately the same as the negative covariance of $\beh-\bet$.
\begin{lem}\label{lem:3}
Under the conditions {\rm (C1)}-{\rm (C3)}, it holds that
\begin{align*}
\E[\Vh]-V=-\E[(\beh-\bet)(\beh-\bet)^t]+O(n^{-5/2}).
\end{align*}
\end{lem}

\begin{proof}
We will show the Lemma by directly comparing both sides of the equation in the Lemma.
Noting that $V=\{X^t(\Si_n+S)X\}^{-1}$ and $A^{-1}-B^{-1}=-A^{-1}(A-B)B^{-1}$ for some non-singular matrices $A$ and $B$, we have
\begin{align*}
\Vh-V=&
-\Vh X^t\{(\Sih_n+S)^{-1}-(\Si_n+S)^{-1}\}XV\\
=&\Vh X^t(\Sih_n+S)^{-1}(\Sih_n-\Si_n)(\Si_n+S)^{-1}XV\\
=&V X^t(\Sih_n+S)^{-1}(\Sih_n-\Si_n)(\Si_n+S)^{-1}XV\\
&+(\Vh-V)X^t(\Sih_n+S)^{-1}(\Sih_n-\Si_n)(\Si_n+S)^{-1}XV\\
\equiv &I_1+I_2.
\end{align*}
Since $V=O(n^{-1})$ and $\Vh-V=O_p(n^{-1/2})$ from the condition (C2), we have 
\begin{align*}
I_1&=V X^t(\Si_n+S)^{-1}(\Sih_n-\Si_n)(\Si_n+S)^{-1}XV\\
&-V X^t(\Si_n+S)^{-1}(\Sih_n-\Si_n)(\Si_n+S)^{-1}(\Sih_n-\Si_n)(\Si_n+S)^{-1}XV\\
&+O_p(n^{-5/2}),
\end{align*}
and
\begin{align*}
I_2=&V X^t(\Si_n+S)^{-1}(\Sih_n-\Si_n)(\Si_n+S)^{-1}XVX^t\\
&\times(\Si_n+S)^{-1}(\Sih_n-\Si_n)(\Si_n+S)^{-1}XV+O_p(n^{-5/2}).
\end{align*}
Then, for $R=X^t(\Si_n+S)^{-1}$ we have
\begin{align}
\E[&\Vh-V]\nonumber\\
=&
\E[VR(\Sih_n-\Si_n)R^tV+VR(\Sih_n-\Si_n)R^tVR(\Sih_n-\Si_n)R^tV\nonumber\\
&\ \ -VR(\Sih_n-\Si_n)(\Si_n+S)^{-1}(\Sih_n-\Si_n)R^tV]+O(n^{-5/2})\nonumber\\
\begin{split}
=&\E[VR(\Sih_n-\Si_n)R^tVR(\Sih_n-\Si_n)R^tV\\
&\ \ \ -VR(\Sih_n-\Si_n)(\Si_n+S)^{-1}(\Sih_n-\Si_n)R^tV]+O(n^{-5/2}),
\end{split}
\label{eqn:eh}
\end{align}
where the last equality holds since $\Sih$ is a second-order unbiased estimator of $\Si$.

Next, we evaluate the first term of the right side of the equation in the Lemma.
We can write $\beh-\bet$ as
\begin{align*}
\beh-\bet
=& (\Vh-V)X^t(\Sih_n +S)^{-1}(y-X\be)
\\
&+V X^t\{(\Sih_n +S)^{-1}-(\Si_n +S)^{-1}\}(y-X\be)\\
=&J_1+J_2.
\end{align*}
In order to approximate the covariance of $\beh-\bet$ up to the order $O(n^{-5/2})$, we expand $J_1$ and $J_2$ as
\begin{align*}
J_1=&
VR(\Sih_n-\Si_n)R^tV X^t(\Si_n +S)^{-1}(y-X\be)+O_p(n^{-1}), \\
J_2=&
-VR(\Sih_n-\Si_n)(\Si_n +S)^{-1}(y-X\be)+O_p(n^{-1}).
\end{align*}
The straightforward calculation shows that
\begin{align*}
\E[J_1J_1^t]&=\E[VR(\Sih_n-\Si_n)R^tV R(\Sih_n-\Si_n)R^tV]+O(n^{-5/2}),\\
\E[J_2J_2^t]&=E[VR(\Sih_n-\Si_n)(\Si_n +S)^{-1}(\Sih_n-\Si_n)R^tV]+O(n^{-5/2}),\\
\E[J_1J_2^t]&=E[J_2J_1^t]=-\E[VR(\Sih_n-\Si_n)R^tVR(\Sih_n-\Si_n)R^tV]+O(n^{-5/2}),
\end{align*}
thereby we have
\begin{align*}
-\E[&(\beh-\bet)(\beh-\bet)^t]\\
=&\E[VR(\Sih_n-\Si_n)R^tVR(\Sih_n-\Si_n)\}R^tV]\\
&-\E[VR(\Sih_n-\Si_n)(\Si_n +S)^{-1}(\Sih_n-\Si_N)R^tV]+O(n^{-5/2}),
\end{align*}
which has the same expression as (\ref{eqn:eh}).
\end{proof}

\section{Proof of Theorem 1}
From Lemma \ref{lem:1}, the conditional distribution of $\beh-\be$ given $Py$  is $N_2 (\beh-\bet, V)$.
Let $w=V^{-1/2}\{(\beh-\be)-(\beh-\bet)\}$.
It is noted that $\Vh-V=O_p(n^{-3/2})$.
Then, the conditional distribution of $w$ given $Py$ is $w \sim N_k (0, I_k)$, and the Mahalanobis' distance is approximated via Taylor series expansion as
\begin{align}
&(\beh-\be)^t\Vh^{-1}(\beh-\be)\nonumber\\
=&
w^tV^{1/2}\Vh^{-1}V^{1/2}w
+2(\beh-\be)^t\Vh^{-1}V^{1/2}w+(\beh-\be)^tV^{-1}(\beh-\bet)\nonumber\\
=&
w^t\Big[I_k-V^{-1/2}(\Vh-V)V^{-1/2}+V^{-1/2}(\Vh-V)V^{-1}(\Vh-V)V^{-1/2}\Big]w\nonumber\\
&+2(\beh-\bet)^t\Vh^{-1}V^{1/2}w+(\beh-\bet)^t\Vh^{-1}(\beh-\bet)
+O_p(n^{-3/2})\nonumber\\
=&
w^t(I_k-G_1)w+2g_2^tw+g_3+O_p(n^{-3/2}),
\label{eqn:mdap}
\end{align}
where 
\begin{align*}
G_1=&
V^{-1/2}(\Vh-V)V^{-1/2}
-V^{-1/2}(\Vh-V)V^{-1}(\Vh-V)V^{-1/2},\\
g_2=&V^{1/2}\Vh^{-1}(\beh-\bet),\\
g_3=&(\beh-\bet)^t\Vh^{-1}(\beh-\bet).
\end{align*}
From (\ref{eqn:mdap}), the characteristic function $\varphi(t)=\E[\exp\{it (\beh-\be)^t\Vh^{-1}(\beh-\be)\}]$ is approximated as 
\begin{align*}
\varphi(t)=&\E\exp\Big(it\{w^t(I_k-G_1)w+2g_2^tw+g_3\}\Big)+O(n^{-3/2})\\
=&
\E\Big[e^{itw^tw} \Big\{ 1+it \{-w^tG_1w+2g_2^tw+g_3\}\\
&-\frac{t^2}2  \{-w^tG_1w+2g_2^tw+g_3\}^2\Big\} \Big]+O(n^{-3/2})\\
=&
\E\Big[e^{itw^tw} \Big\{ 1+it \{-w^tG_1w+2g_2^tw+g_3\}
\\
&-\frac{t^2}2  \{(w^tG_1w)^2+4w^tg_2g_2^tw-4w^tG_1wg_2^tw\}
\Big\} \Big]+O(n^{-3/2}),
\end{align*}
because $G_1=O_p(n^{-1/2})$, $g_2=O_p(n^{-1/2})$ and $g_3=O_p(n^{-1})$.
From the law of iterated expectations and the conditional normality of $w$, the above equation reduces to
\begin{align*}
\varphi(t)=&
\E\Big[e^{itw^tw} \Big\{ 1+it \{-w^tG_1w+g_3\}\\
&
\qquad\qquad-\frac{t^2}2  \{(w^tG_1w)^2
+4w^tg_2g_2^tw\}\Big\} \Big]
+O(n^{-3/2}).
\end{align*}
For some deterministic matrix $A$ and $w \sim N_k (0, I_k)$, it holds that
\begin{align*}
\E\Big[e^{itw^tw} w^t A w \Big]=&
(2\pi)^{-k/2}\int \exp\Big(-\frac{(1-2it)w^tw}2\Big)w^t A w dw\\
=&
(1-2it)^{-k/2-1} \tr (A),\\
\E\Big[e^{itw^tw} (w^t A w)^2 \Big]=&
(2\pi)^{-k/2}\int \exp\Big(-\frac{(1-2it)w^tw}2\Big)(w^t A w)^2 dw\\
=&
(1-2it)^{-k/2-2} (\tr ^2(A)+2\tr (A^2)).
\end{align*}
Using these equalities, from the law of iterated expectations, we have
\begin{align*}
\varphi(t)=&
(1-2it)^{-k/2} \Big[ 1+it \Big\{-(1-2it)^{-1}\tr(\E[G_1])+E[g_3]\Big\}\\
&\qquad+\frac{(it)^2}2 \Big\{(1-2it)^{-2} \{\E[\tr^2(G_1)]+2\tr(\E[G_1^2]) \}
\\
&\qquad\qquad+(1-2it)^{-1}4\tr(\E[g_2g_2^t])\Big\}
\Big]+O(n^{-3/2}).
\end{align*}
For notational simplicity, let $J=\E[\tr^2(G_1)]+2\tr(\E[G_1^2])$.
Let $s= (1-2it)^{-1}$, or $it=(s-1)/(2s)$.
Then,  $(1-2it)^{k/2}\varphi(t)-1$ can be written as
\begin{align}
&it \Big\{-(1-2it)^{-1}\tr(E[G_1])+E[g_3]\Big\}+{(it)^2 \over 2}  \Big\{(1-2it)^{-2} J 
+(1-2it)^{-1}4E[g_2^tg_2]\Big\}\nonumber\\
\begin{split}
=&
\frac1{2s}\Big\{E[g_2^tg_2]-E[g_3]\Big\}+ \Big\{\frac12\tr(E[G_1])+\frac12E[g_3] +\frac J8- E[g_2^tg_2]\Big\}
\\
&+\Big\{ -\frac12\tr(E[G_1])-\frac J4
+\frac12E[g_2^tg_2]\Big\}s
+\frac J8 s^2
\label{eqn:pol}
\end{split}
\end{align}

We shall evaluate the moments in (\ref{eqn:pol}).
First, $G_1$ can be expressed as 
\begin{align}
\begin{split}
G_1
=&
V^{-1/2}(\Vh-V)V^{-1/2}-V^{-1/2}(\Vh-V)V^{-1}(\Vh-V)V^{-1/2},
\end{split}
\label{eqn:G12ap}
\end{align}
thereby it holds that 
\begin{align}
\tr(E[G_1])= \tr(E[K])-\tr(\E[K^2]),
\label{eqn:pcr1}
\end{align}
for $K=V^{-1/2}(\Vh-V)V^{-1/2}$.
Noting that the first term in (\ref{eqn:G12ap}) is $O_p(n^{-1/2})$ and the second term is $O(n^{-1})$, we can expand $G_1^2$ and $\tr^2(G_1)$ as
\begin{align*}
&G_1^2
=
V^{-1/2}(\Vh-V)V^{-1}(\Vh-V)V^{-1/2}+O_p(n^{-3/2}),\\
&\tr^2(G_1)
=
\tr^2(V^{-1/2}(\Vh-V)V^{-1/2})+O_p(n^{-3/2}),
\end{align*}
which lead to $\E[G_1^2]=\E[K^2] +O(n^{-3/2})$ and $\E[\tr^2(G_1)]=\E[\tr^2(K)]+O(n^{-3/2})$.
Thus,
\begin{equation}
J=\E[\tr^2(K)]+2\tr(\E[K^2])+O(n^{-3/2}).
\label{eqn:pcr2}
\end{equation}
It can be also observed that
\begin{align*}
&g_2^tg_2
=
(\beh-\bet)^t\Vh^{-1}V\Vh^{-1}(\beh-\bet)
=
(\beh-\bet)^tV^{-1}(\beh-\bet)+O_p(n^{-3/2}),\\
&g_3
=(\beh-\bet)^t\Vh^{-1}(\beh-\bet)
=
(\beh-\bet)^tV^{-1}(\beh-\bet)+O_p(n^{-3/2}).
\end{align*}
Then, from Lemma \ref{lem:3} we have
\begin{equation}
E[g_2^tg_2]=E[g_3]
=-\tr(E[K])+O(n^{-3/2}).
\label{eqn:pcr3}
\end{equation}

Combining (\ref{eqn:pcr1}), (\ref{eqn:pcr2}) and (\ref{eqn:pcr3}), we can see that the characteristic function of $(\beh-\bet)^t\Vh^{-1}(\beh-\bet)$ can be written as
\begin{align*}
\varphi(t)=&(1-2it)^{-k/2}\{1+B_1/8-B_2/4+B_3+(-B_1/4-B_3)s+(B_1/8+B_2/4)s^2\}\\
&+O(n^{-3/2}),
\end{align*}
for $B_1, B_2$ and $B_3$ are defined in the main document.
From the fact that the characteristic function of the chi-squared distribution with degrees of freedom $k+2h$ is given by $(1-2it)^{-k/2-h}=(1-2it)^{-k/2}s^h$, it follows that the asymptotic expansion of the cumulative distribution function of $(\beh-\be)^t\Vh^{-1}(\beh-\be)$ is 
\begin{align*}
F_{k}&(x)+(B_1/8-B_2/4+B_3)F_{k}(x)\\
&+(-B_1/4-B_3)F_{k+2}(x)+(B_1/8+B_2/4)F_{k+4}(x)+O(n^{-3/2}),
\end{align*}
where $F_{k}(x)$ is the cumulative distribution function of the chi-squared distribution with degrees of freedom $k$.
Note that $F_{k+r-2}(x)-F_{k+r}(x)=2f_{k+r}(x)$, where $f_k(x)$ is the density function of the chi-squared distribution with degrees of freedom $k$.
Then, it holds that 
\begin{equation*}
\begin{split}
P(&(\beh-\be)^t\Vh^{-1}(\beh-\be)\leq x)\\
=&F_{k}(x)+2\left(\frac{B_1}{8}-\frac{B_2}{4}+B_3\right)f_{k+2}(x)-\left(\frac{B_1}{4}+\frac{B_2}{2}\right)f_{k+4}(x)+O(n^{-3/2}),
\end{split}
\end{equation*}
thereby, for a function $h=h(\Si)$ with order $O(n^{-1})$, we have
\begin{align*}
&P\{(\beh-\be)^t\Vh^{-1}(\beh-\be)\leq x(1+h)\}\\
=&F_{k}(x)+hxf_k(x)+\left(\frac{B_1}{4}-\frac{B_2}{2}+2B_3\right)f_{k+2}(x)-\left(\frac{B_1}{4}+\frac{B_2}{2}\right)f_{k+4}(x)+O(n^{-3/2}),
\end{align*}
which completes the proof.

\section{Derivation of the equation (6)}

We write functions given in Section 2 as functions of $\Si$ since the unknown parameter is $\Si$ in this example.
For $V=(\sum_{i=1}^nD_i^{-1})^{-1}$ and $D_i=\Si+S_i$, $\Vh-V$ can be expanded as
\begin{align*}
\Vh-V
=
V\Big\{\sum_{i=1}^nD_i^{-1}(\Sih-\Si)D_i^{-1}\Big\}V+O_p(n^{-2}).
\end{align*}
Since the first term on the right side of the above equation is of order $O_p(n^{-3/2})$, we only need to consider this term to derive the expressions given in (6).

At first, we evaluate $B_1^\ast$.
It is noted that we have
\begin{align*}
&V^{-1/2}(\Vh-V)V^{-1/2}\\
=&\frac1n\sum_{j=1}^n V^{1/2}\Big\{\sum_{i=1}^nD_i^{-1}D_j^{1/2}\{u_ju_j^t-I_p\}D_j^{1/2}D_i^{-1}\Big\}V^{1/2}+O_p(n^{-1}),
\end{align*}
where $u_j$ are independently distributed as the standard normal distribution.
Then, we have
\begin{align}
B_1^\ast=\E\Big[&\tr^2\Big\{V^{-1/2}(\Vh-V)V^{-1/2}\Big\}\Big]\nonumber\\
=&\frac1{n^2}\sum_{i=1}^n\sum_{j=1}^n\sum_{k=1}^n \E\Big[\tr\Big\{VD_j^{-1}D_i^{1/2}\{u_iu_i^t-I_p\}(\Si+S_i)^{1/2}D_j^{-1}\Big\}\nonumber\\
&\times \tr\Big\{VD_k^{-1}D_i^{1/2}\{u_iu_i^t-I_p\}(\Si+S_i)^{1/2}D_k^{-1}\Big\}\Big]\nonumber\\
\begin{split}
=&
\frac2{n^2}\sum_{i=1}^n\sum_{j=1}^n\sum_{k=1}^n \tr\Big[D_i^{1/2}D_j^{-1}VD_j^{-1}D_i^{1/2}D_i^{1/2}D_k^{-1}VD_k^{-1}D_i^{1/2}\Big].
\label{eqn:a1}
\end{split}
\end{align}

Next, we evaluate $B_2^\ast$.
It is noted that we have 
\begin{align*}
(\Vh&-V)V^{-1}(\Vh-V)V^{-1}\\
=&
V\Big\{\sum_{i=1}^nD_i^{-1}(\Sih-\Si)D_i^{-1}\Big\}V\Big\{\sum_{i=1}^nD_i^{-1}(\Sih-\Si)D_i^{-1}\Big\}+O_p(n^{-3/2})
\end{align*}
and that $\Sih-\Si$ can be written as
\begin{align*}
\Sih-\Si
=
\frac1n \sum_{i=1}^n \{(y_i-\be)(y_i-\be)^t-\Si -S_i\}+O_p(n^{-1}).
\end{align*}
Then, for $u_i$ for $i=1,\ldots,n$ which are independently distributed as the multivariate standard normal distribution, it holds that for $\ell, m=1,\ldots,n$,
\begin{align*}
\E\Big[&
(\Sih-\Si)D_\ell^{-1}VD_m^{-1}(\Sih-\Si)\Big]\\
=&
\frac1{n^2} \sum_{i=1}^n\sum_{j=1}^nD_i^{1/2}\E\Big[(u_iu_i^t-I_p)(\Si+S_i)^{1/2}D_\ell^{-1}VD_m^{-1}D_j^{1/2}(u_ju_j^t-I_p)\Big]D_j^{1/2}\\
=&
\frac1{n^2} \sum_{i=1}^nD_i^{1/2}\E\Big[(u_iu_i^t-I_p)(\Si+S_i)^{1/2}D_\ell^{-1}V(\Si)D_m^{-1}D_i^{1/2}(u_iu_i-I_p)\Big]D_i^{1/2}\\
=&
\frac1{n^2} \sum_{i=1}^nD_i^{1/2}(L_{i\ell m}+\tr(L_{i\ell m})I_p)D_i^{1/2}
\end{align*}
for $L_{i\ell m}=D_i^{1/2}D_\ell^{-1}VD_m^{-1}D_i^{1/2}$.
Then, we have
\begin{equation}
\begin{split}
B_2^\ast=&\E[\tr(\{(\Vh-V)V^{-1}\}^2)]\\
=&\frac1{n^2}\sum_{i,\ell,m=1}^n\tr\Big(D_\ell^{-1}D_i^{1/2}L_{i\ell m}D_i^{1/2}D_m^{-1}V\Big)+\frac1{n^2} \sum_{i,\ell,m=1}^n\tr(L_{i\ell m})\tr\Big(D_\ell^{-1}D_iD_m^{-1}V\Big)\\
=&\frac1{n^2} \sum_{i=1}^n\tr\Big\{V\sum_{j=1}^n\Big(D_j^{-1}D_iD_j^{-1} \Big)^2\Big\}+\frac1{n^2} \sum_{i,\ell,m=1}^n\tr^2\Big(D_i^{-1}D_jD_k^{-1}V(\Si)\Big).
\label{eqn:a2}
\end{split}
\end{equation}

Finally, we evaluate $B_3^\ast$.
From the equation (\ref{eqn:eh}), for $V=(\sum_{i=1}^nD_i^{-1})^{-1}$ we have 
\begin{align*}
\E[K]=&
\E\Big[
V^{1/2}\Big\{\sum_{i=1}^nD_i^{-1}(\Sih-\Si)D_i^{-1}\Big\}V\Big\{\sum_{i=1}^nD_i^{-1}(\Sih-\Si)D_i^{-1}\Big\}V^{1/2}\\
&\ \ \ -V^{1/2}\Big\{\sum_{i=1}^nD_i^{-1}(\Sih-\Si)D_i^{-1}(\Sih-\Si)D_i^{-1}\Big\}V^{1/2}\Big]+O(n^{-3/2}).
\end{align*}
The trace of the first term in the above equation is exactly the same with $B_2^\ast$ and is given in (\ref{eqn:a2}).
To evaluate the second term, it is noted that
\begin{align*}
\E\Big[&\sum_{i=1}^nD_i^{-1}(\Sih-\Si)D_i^{-1}(\Sih-\Si)D_i^{-1}\Big]\\
=&\frac1{n^2} \sum_{i=1}^n\sum_{j=1}^n D_i^{-1}D_jD_i^{-1}D_jD_i^{-1}+\frac1{n^2} \sum_{i=1}^n\sum_{j=1}^n \tr(D_i^{-1}D_j)D_i^{-1}D_jD_i^{-1}.
\end{align*}
Then, the trace of the second term in the above equation is given by
\begin{align}
-\tr&\Big(\E\Big[V^{1/2}\Big\{\sum_{i=1}^nD_i^{-1}(\Sih-\Si)D_i^{-1}(\Sih-\Si)D_i^{-1}\Big\}V^{1/2}\Big]\Big)\nonumber\\
=&-\frac1{n^2} \sum_{i=1}^n\sum_{j=1}^n \tr\Big(VD_i^{-1}D_jD_i^{-1}D_jD_i^{-1}\Big)-\frac1{n^2} \sum_{i=1}^n\sum_{j=1}^n \tr\Big(D_i^{-1}D_j\Big)\tr\Big(VD_i^{-1}D_jD_i^{-1}\Big).
\label{eqn:a3}
\end{align}
Equation (\ref{eqn:a1}), (\ref{eqn:a2}) and (\ref{eqn:a3}) lead to the expression given in (\ref{B123}) in the main document.


\end{document}